\documentclass[12pt]{article}

\usepackage[hmargin=1.3in,vmargin=1.3in]{geometry}
\usepackage{multirow}

\usepackage{amsmath,amsthm,amssymb,amsfonts,verbatim,enumerate}

\numberwithin{equation}{section}

\RequirePackage[colorlinks,pagebackref,linkcolor=blue,filecolor = blue,citecolor = blue, urlcolor =blue]{hyperref}

\newtheorem{lemma}{Lemma}[section]
\newtheorem{theorem}[lemma]{Theorem}

\newtheorem{example}[lemma]{Example}

\newcommand{\be}{\begin{equation}}\newcommand{\ee}{\end{equation}}
\newcommand{\bes}{\begin{equation*}}\newcommand{\ees}{\end{equation*}}
\newcommand{\ba}{\begin{array}}\newcommand{\ea}{\end{array}}
\newcommand{\ben}{\begin{eqnarray}}\newcommand{\een}{\end{eqnarray}}
\newcommand{\bn}{\begin{eqnarray*}}\newcommand{\en}{\end{eqnarray*}}

\theoremstyle{remark}
\newtheorem{rmk}{Remark}
\begin{document}
\title{The parameters of a family of linear codes}
\author{\small Weiqiong Wang$^1$\thanks{wqwang@chd.edu.cn}, Yan Wang$^2$\thanks{Corresponding author: lanse-wy@163.com}\\
\small 1. School of Science, Chang'an University, Xi'an, 710064, China\\
\small 2. School of Science, Xi'an University of Architecture and
Technology, Xi'an, 710055, China}
\date{}
\maketitle
\thispagestyle{empty}
\begin{abstract}
  A large family of linear codes with flexible parameters from almost bent functions and perfect nonlinear functions are constructed and their parameters  are  determined. Some constructed linear codes  and their related codes are optimal in the sense that they meet certain bounds on linear codes. Applications of these codes in secret sharing schemes and combinational designs are presented.

\end{abstract}
\section{Introduction}

Let $p$ be a prime and $q=p^m$ for some positive integer $m$. An $[n,k,d]$ linear code $\mathcal{C}$ over $\mbox{GF}(p)$ is a $k-$dimensional subspace of $\mbox{GF}(p)^n$ with Hamming distance $d$. Let $A_i(0\leq i\leq n)$ denote the number of codewords with Hamming weight $i$ in a code $\mathcal{C}$. The \emph{weight distribution} $(A_0, A_1, \cdots,  A_n)$
is important in coding theory as it contains important information about the error-correcting capability and the probability of error detection and correction. A code $\mathcal{C}$ is said to be a $t-$weight code if the number of nonzero $A_i$ in the sequence $(A_1, \cdots,  A_n)$ is equal to $t$, and the \emph{weight enumerator} of $\mathcal{C}$ is defined  as
$$A(z)=A_0+A_1z+A_2z^2+\cdots +A_nz^n.$$

Linear codes are widely studied because they have applications in computer and communication systems, data storage devices and consumer electronics. There are several approaches to constructing linear codes. One of which is based on  functions over finite fields. In the past decade, much progress has been made on interplay between functions and codes. Some special types of functions like bent functions, semi-bent functions, almost bent functions, APN functions, planar functions, and Dickson polynomials were employed to construct linear codes with good parameters.
For example, bent functions were extensively employed to construct linear codes with few weight \cite{Ding2015, Mesnager2015code, Tang2016,Wolfmann99, Zhou2015}. This kind of linear codes have applications in secrete sharing \cite{Anderson,CarletDing05,kDing2015,YUanDing2006}, authentication codes \cite{Dingwang2005}, association scheme \cite{Calderbank84},  combinational designs \cite{Ding2017}, and strongly regular graphs \cite{Calderbank86}.

In this paper, we focus on the following generic construction of linear codes:

 \begin{equation}\label{orgdef}
\mathcal{C}_g=\{\big(\mbox{Tr}_1^m(ag(x)+bx)\big)_{x\in \mbox{GF}(p^m)^*}: a,b\in \mbox{GF}(p^m)\},
\end{equation}
where $g$ is a function from $\mbox{GF}(p^m)$ to $\mbox{GF}(p^m)$, $\mbox{Tr}_1^m$ is the trace function from $\mbox{GF}(p^m)$ to $\mbox{GF}(p)$.

Perfect nonlinear functions and almost bent functions have been employed via this approach to construct linear codes in the literature \cite{CarletDing05,Ding2007,Ding2017,DingXiang2015,DingYuan06,Fengluo07}. The constructed codes were proved to have good parameters and could be used to construct secrete sharing scheme and authenticated codes. The main objective of this paper is to generalize this construction by restricting the parameter $a$ in $(\ref{orgdef})$ to certain additive subgroup of $(\mbox{GF}(p^m), +)$ with order $p^r$, where $0\leq r\leq m$, which of course includes the  group $(\mbox{GF}(p^m), +)$ as a special case. We then determine the parameters of the constructed codes  as well as parameters of their related codes. As will be seen, a large family of  the constructed codes are  optimal or at least almost optimal, and can be used in many applications.

Our motivation of this paper is two-fold. From the  theoretical point of view, we can obtain much more good linear codes. Our method does not change the  parameters $n$ and $d$, but the dimension changes to $m+r$ ($1\leq r \leq m$), which may provide more optimal  linear codes. Although our new codes with parameters $[n, m+r, d]$ $(1<r<m)$ are not as good as the original ones with parameters $[n, 2m, d]$, some of them are still distance optimal or almost optimal. Moreover, the dual codes of ours may be better than dual codes  of the original ones. Since the dimension of our dual codes can be much larger than the original ones while the minimum distance does not change much. From the  application point of view, we can obtain more flexible  secrete sharing schemes with high democracy as the access structures derived from our  codes are much more flexible.

\section{Some auxiliary results}
In this section, we  introduce some
auxiliary results which will be required in later sections.

\subsection{Highly nonlinear Boolean functions}

Let $f$ be a Boolean function from $\mbox{GF}(2^m)$ to $\mbox{GF}(2)$.
The \emph{Walsh transform} of $f$ at $w\in \mbox{GF}(2^m)$ is defined by
\begin{equation}
\widehat{f}(w)=\sum_{x\in \text{GF}(2^m)}(-1)^{f(x)+\text{Tr}_1^m(wx)},
\end{equation}
where $\mbox{Tr}_1^m$ is the absolute trace function from $\text{GF}(2^m)$ to $\text{GF}(2)$.
 The \emph{Walsh spectrum} of $f$ is the following multiset
\begin{equation}
\{\widehat{f}(w): w\in \text{GF}(2^m)\}.
\end{equation}

A function $f$ from $\text{GF}(2^m)$  to $\text{GF}(2)$ is called \emph{linear} if $f(x+y)=f(x)+f(y)$ for all $(x, y)\in \text{GF}(2^m)^2$.  It is called \emph{affine} if $f$ or $f-1$ is linear.

A function $f$ from $\text{GF}(2^m)$  to $\text{GF}(2)$ is called \emph{bent} if $|\widehat{f}(w)|=2^{m/2}$ for every $w\in \text{GF}(2^m)$. Bent functions exist only for even $m$, and were coined by Rothaus in \cite{Rothaus1976}.

A subset $D$ with $\kappa$ elements of an abelian group $(A, +)$ of order $\nu$ is called an $(\nu, \kappa, \lambda)$ \emph{difference set} in $(A, +)$ if the multiset $\{x-y: x\in D, y\in D\}$ contains every nonzero element of $A$ exactly $\lambda$ times. It is well-known that a function $f$ from $\text{GF}(2^m)$ to $\text{GF}(2)$ is bent if and only if $D_f$ is a Hadamard difference set in $(\text{GF}(2^m), +)$ with the following parameters
\begin{equation}
(2^m, 2^{m-1}\pm 2^{(m-2)/2}, 2^{m-2}\pm 2^{(m-2)/2}).
\end{equation}

Let $m$ be odd. Then there is no bent function on $\text{GF}(2^m)$. A function $f$ from $\text{GF}(2^m)$ to $GF(2)$ is called \emph{semi-bent} if $\widehat{f}(w)\in \{0, \pm 2^{(m+1)/2}\}$ for every $w\in \text{GF}(2^m)$.

\subsection{Almost bent functions}

For any function $g$ from $\text{GF}(2^m)$ to $\text{GF}(2^m)$, we define

$$\lambda_g(a,b)=\sum_{x\in GF(2^m)}(-1)^{\text{Tr}_1^m(ag(x)+bx)}, \ \ a, b\in \text{GF}(2^m).$$

A function $g$ from $\text{GF}(2^m)$ to $\text{GF}(2^m)$ is called almost bent if $\lambda_g(a, b)=0,$ or $\pm 2^{(m+1)/2}$ for every pair $(a, b)$ with $a\neq 0$.

By definition, almost bent functions over $\text{GF}(2^m)$ exist only for odd $m$. Any almost bent function $g(x)$ provides a set of $2^m-1$ semi-bent functions
$$ \{\text{Tr}_1^m(ag(x)): a\in \text{GF}(2^m)\setminus \{0\}\}.$$

The following is a list of almost bent functions on $\text{GF}(2^m)$, where $m$ is odd.

$1) \ \ g(x)=x^{2^i+1},\ $gcd$(i, m)=1$.

$2) \ \ g(x)=x^{2^{2i}-2^i+1},\ $gcd$(i, m)=1$.

$3) \ \ g(x)=x^{2^{(m-1)/2}+3}$.

$4) \ \ g(x)=x^{2^{(m-1)/2}+2^{(m-1)/4}-1},\ m\equiv 1 ($mod $4)$.

$5) \ \ g(x)=x^{2^{(m-1)/2}+2^{(3m-1)/4}-1},\ m\equiv 3 ($mod$ 4)$.

$6) \ \ g(x)=x^{2^i+1}+(x^{2^i+1}+x)$Tr$_1^m(x^{2^i+1}+x),\ m>3,\ $gcd$(i, m)=1$.

\subsection{Planar functions}\label{sec:planar}
Let $p$ be an odd prime. A function $f$ from $\text{GF}(p^m)$ to itself is called a \emph{planar function} if the difference function $f_a(x)=f(x+a)-f(x)$ is a one-to-one function from $\text{GF}(p^m)$ to itself for every $a\in \text{GF}(p^m)^*$.

The following is a list of some known planar functions on $\text{GF}(p^m)$ $\cite{Carlet2004, Coulter1997, Coulter2002}$:

$1)$ $f_1(x)=x^{p^t+1}$, where $t\geq 0$ is an integer and $m/gcd(m, t)$ is odd (Dembowski-Ostrom \cite{Dembowski1968}, including the function $x^2$ as a special case).

$2)$ $f_2(x)=x^{(3^k+1)/2}$, where $p=3, k$ is odd, and $gcd(m, k)=1$ (Coulter-Matthews \cite{Coulter1997}).

$3)$ $f_3(x)=x^{10}-ux^6-u^2x^2$, where $p=3, $ $m$ is odd, and $u\in \mbox{GF}(p^m)^*$ (Coulter-Matthews $\cite{Coulter1997}$, Ding-Yuan  $\cite{DingYuan}$).

\subsection{The MacWilliams Identity}

Assume $\mathcal{C}^\perp$ is the dual of an $[n, k, d]$ linear code $\mathcal{C}$ over $\text{GF}(p)$. Denote by $A(z)$ and $A^\perp(z)$ the weight enumerators of $\mathcal{C}$ and $\mathcal{C}^\perp$ respectively. The  MacWilliams Identity (\cite{Van99}, p. 41) shows that  $A(z)$ and $A^\perp(z)$ can be derived from each other as follows.

\begin{theorem}
let $C$ be an $[n,k,d]$ code over $\text{GF}(p)$ with weight enumerator $A(z)=\sum_{i=0}^nA_iz^i$, and let $A^\perp(z)$ be the weight enumerator of $C^\perp$. Then
\begin{equation}
A^\perp(z)=p^{-k}(1+(p-1)z)^nA\bigg(\frac{1-z}{1+(q-1)z}\bigg).
\end{equation}
\end{theorem}

\section{Binary linear codes from almost bent functions}

In this section, we derive  a family of linear codes from almost bent functions, and analyze their parameters as well as the parameters of  their related codes.

For any function $g$ from $\text{GF}(2^m)$ to $\text{GF}(2^m)$ with $g(0)=0$, we define the following linear code
\begin{equation}\label{def:1}
\mathcal{C}_{(g, A)}=\{(\text{Tr}_1^m(ag(x)+bx))_{x\in \text{GF}(2^m)^*}: a\in A, b\in \text{GF}(2^m)\},
\end{equation}
where $A$ is an additive subgroup of $(\text{GF}(2^m), +)$ with order $2^r$,  $0\leq r \leq m$.

For any fixed $a\in A, b\in \text{GF}(2^m)$,  we denote the corresponding codeword by
$$c_{a, b}=(\text{Tr}_1^m(ag(x)+bx))_{x\in \text{GF}(2^m)^*}.$$

When $g$ is an almost bent function and $A=\text{GF}(2^m)$, the parameters and weight distribution of the code $\mathcal{C}_{(g, A)}$ are known in the literature \cite{Ding2017, DingXiang2015}.

We determine in this section the parameters and weight distribution of the code  $\mathcal{C}_{(g, A)}$ for  the case of $g$ being an almost bent function, and $A$ being any additive subgroup of $(\text{GF}(2^m), +)$.

 \subsection{The parameters of $\mathcal{C}_{(g,A)}$ }
\begin{theorem}\label{thm:1}
Let $m\geq 3$ be an odd  integer. Let $g$ be an almost bent function from $\text{GF}(2^m)$ to $\text{GF}(2^m)$ with $g(0)=0$. Let $A$ be an additive subgroup of $(\text{GF}(2^m), +)$ with order $2^r$, where $0\leq r \leq m$. Then the code $\mathcal{C}_{(g, A)}$ defined in $(\ref{def:1})$ has parameters $[2^m-1, m+r, 2^{m-1}- 2^{(m-1)/2}]$ with weight distribution in Table $\ref{table:1}$.

\begin{table}[h]
\caption{Weight distribution of $\mathcal{C}_{(g,A)}$}
\centering

\begin{tabular}{ |l|l| }

\hline

\hline   Weight $w$ & Multiplicity $A_w$ \\
\hline

$0$ & $1$ \\

$2^{m-1}- 2^{(m-1)/2}$ & $(2^r-1)(2^{m-2}+2^{\frac{m-3}{2}})$ \\

$2^{m-1}$ & $2^{m-1}(2^r+1)-1$ \\

$2^{m-1}+ 2^{(m-1)/2}$ & $(2^r-1)(2^{m-2}-2^{\frac{m-3}{2}})$  \\
\hline
\end{tabular}
\label{table:1}
\end{table}
\end{theorem}
\begin{proof}

Firstly, the length of code $\mathcal{C}_{(g, A)}$ is $2^m-1$ according to its definition in $(\ref{def:1})$.

We consider then the Hamming weight of any fixed codeword $c_{a, b}$ in  $\mathcal{C}_{(g, A)}$, and denote it by $wt(c_{a, b})$.

If $a=b=0$, $c_{a, b}$ is the zero codeword.

If $a=0$, $b\neq 0$, then $wt(c_{a, b})=2^{m-1}$ since $\text{Tr}_1^m(bx)$ is linear and nonnull.

If $a\neq 0$,  put $N_{a,b}=\sharp\{x\in \text{GF}(2^m)^*: \text{Tr}_1^m(ag(x)+bx)=0\}$, then
 \begin{eqnarray*} 2N_{a,b}&= &\sum_{x\in \text{GF}(2^m)^*}\sum_{y\in \text{GF}(2)}(-1)^{y\cdot \text{Tr}_1^m(ag(x)+bx)}\\
 &=& \sum_{x\in \text{GF}(2^m)}\sum_{y\in \text{GF}(2)}(-1)^{y\cdot \text{Tr}_1^m(ag(x)+bx)}-2\\
&=& 2^m-2+\sum_{x\in \text{GF}(2^m)}(-1)^{ \text{Tr}_1^m(ag(x)+bx)}.\\
\end{eqnarray*}

Since $g$ is an almost bent function from  $\text{GF}(2^m)$ to $\text{GF}(2^m)$, we know from its definition that $\sum_{x\in GF(2^m)}(-1)^{\text{Tr}_1^m(ag(x)+bx)}\in \{0,\pm 2^{\frac{m+1}{2}}\}$.
So
$$2N_{a, b}\in  \{2^m-2,2^m-2+2^{\frac{m+1}{2}}, 2^m-2-2^{\frac{m+1}{2}}\}.$$

As a result, we have
 \begin{eqnarray*}
  wt(c_{a, b})&=&2^m-1-N_{a,b}\\
&\in & \{2^m-2^{m-1}, 2^m-2^{m-1}- 2^{\frac{m-1}{2}}, 2^m-2^{m-1}+ 2^{\frac{m-1}{2}}\}\\
&= & \{2^{m-1}, 2^{m-1}- 2^{\frac{m-1}{2}}, 2^{m-1}+ 2^{\frac{m-1}{2}}\}.
\end{eqnarray*}

 From the analysis above, we know $c_{0,0}$ is the unique zero codeword of $\mathcal{C}_{(g, A)}$, which implies $\mathcal{C}_{(g, A)}$ has exactly $2^{m+r}$ distinct codewords. On the other hand,  $\mathcal{C}_{(g, A)}$ is obviously linear. Hence $\mathcal{C}_{(g,A)}$ is a linear code with dimension  $m+r$. Furthermore, the codeword $c_{a, b}$ has the following possible nonzero weights:

  \begin{equation}
   \left\{
   \begin{aligned}
w_1&=2^{m-1}- 2^{(m-1)/2},\\
w_2&=2^{m-1},\\
w_3&=2^{m-1}+ 2^{(m-1)/2}.
\end{aligned}
\right.
\end{equation}

 What we should do next is to determine the multiplicity of each weight.

 It is not hard to know $A_1^\perp=A_2^\perp=0$. Because there is no $x\in \text{GF}(2^m)^*$ such that $\text{Tr}_1^m(ag(x)+bx)=0 $ holds for all $a\in A, b\in \text{GF}(2^m)$. There are no pair of $(x, y)$($x\neq y \in \text{GF}(2^m)^*$), such that $\text{Tr}_1^m(ag(x)+bx)=\text{Tr}_1^m(ag(y)+by)$ holds for all $a\in A, b\in \text{GF}(2^m)$ too.
 So by the first three Pless power moments (\cite{Huffman03}, p. 260), we have
\begin{equation}
   \left\{
   \begin{aligned}
1+A_{w_1}+A_{w_2}+A_{w_3}&=2^{m+r},\\
w_1A_{w_1}+w_2A_{w_2}+w_3A_{w_3}&=2^{m+r-1}(2^m-1),\\
w_1^2A_{w_1}+w_2^2A_{w_2}+w_3^2A_{w_3}&=2^{2m+r-2}(2^m-1).
\end{aligned}
\right.
\end{equation}
The weight distribution of $\mathcal{C}_{(g, A)}$ in Table \ref{table:1} then follows from solving this system of linear equations.
\end{proof}

 \subsection{Parameters of the dual of $\mathcal{C}_{(g, A)}$ }
 Based on the  weight distribution of $\mathcal{C}_{(g, A)}$,  its weight enumerator can be written as
 \begin{equation}\label{weinub:1}
 A(z)=1+A_{w_1}z^{2^{m-1}-2^{(m-1)/2}}+A_{w_2}z^{2^{m-1}}+A_{w_3}z^{2^{m-1}+2^{(m-1)/2}},
 \end{equation}
 where $A_{w_1}$, $A_{w_2}$, $A_{w_3}$ are given in Table \ref{table:1}.

 Due to the MacWilliams Identity, the parameters of the dual of $\mathcal{C}_{(g, A)}$ can also be determined.

\begin{theorem}\label{thm:2}
Let $m\geq 3$ be an odd integer. Let  $\mathcal{C}_{(g, A)}$ be a binary code with parameters $[2^m-1, m+r, 2^{m-1}- 2^{(m-1)/2}]$ and weight distribution in Table $\ref{table:1}$. Then its dual code  $\mathcal{C}_{(g, A)}^\perp$ has parameters $[2^m-1, 2^m-1-m-r, d^\perp]$, and its weight distribution is given by
$$ 2^{m+r}A_k^\perp(z)=\binom{2^m-1}{k}+A_{w_1}U_1(k)+A_{w_2}U_2(k)+A_{w_3}U_3(k),$$
where $0\leq k \leq 2^m-1$,
 \begin{equation}\label{U:1}
   \left\{
   \begin{aligned}
U_1(k)&=\sum_{i+j=k}(-1)^i\binom{w_1}{i}\binom{2^m-1-w_1}{j},\\
U_2(k)&=\sum_{i+j=k}(-1)^i\binom{w_2}{i}\binom{2^m-1-w_2}{j},\\
U_3(k)&=\sum_{i+j=k}(-1)^i\binom{w_3}{i}\binom{2^m-1-w_3}{j}.
\end{aligned}
\right.
\end{equation}
Furthermore,
\begin{eqnarray}
     d^\perp  &= & \left\{\begin{array}{ll}
5 & {\rm if}\ \ \ r=m,\\
3 & {\rm if}\ \ \ r\neq m.
\end{array}
\right.
\end{eqnarray}
\end{theorem}

\begin{proof}
Note that the weight enumerator of $\mathcal{C}_{(g, A)}$ is given in $(\ref{weinub:1})$.
It then follows from the MacWilliams Identity that the weight enumerator of $\mathcal{C}_{(g, A)}^\perp$ is given by
 \begin{eqnarray*} 2^{m+r}A^\perp(z)&= &(1+z)^{2^m-1} A(\frac{1-z}{1+z})\\
 &=& (1+z)^{2^m-1}\left[1+A_{w_1} \frac{(1-z)^{w_1}}{(1+z)^{w_1}}+A_{w_2} \frac{(1-z)^{w_2}}{(1+z)^{w_2}}+A_{w_3} \frac{(1-z)^{w_3}}{(1+z)^{w_3}}\right]\\
&=&(1+z)^{2^m-1}+A_{w_1} (1-z)^{w_1}(1+z)^{2^m-1-w_1}\\
&&+A_{w_2} (1-z)^{w_2}(1+z)^{2^m-1-w_2}+A_{w_3}(1-z)^{w_3}(1+z)^{2^m-1-w_3}.
\end{eqnarray*}

One can easily see that the value of $2^{m+r}A_k^\perp$ is exactly the coefficient of $z^k$ on the right hand side of the above equation.
That is
\begin{equation}\label{dual:1}
2^{m+r}A_k^\perp=\binom{2^m-1}{k}+A_{w_1}U_1(k)+A_{w_2}U_2(k)+A_{w_3}U_3(k),
\end{equation}
where $U_1(k)$, $U_2(k)$ and $U_3(k)$ are defined in $(\ref{U:1})$.

We know from Theorem $\ref{thm:1}$  that the dimension of $\mathcal{C}_{(g, A)}$ is $m+r$. Therefore, the
dimension of $\mathcal{C}_{(g, A)}^\perp$  is  $2^m-1-m-r$. Finally, we prove  the minimum distance of $\mathcal{C}_{(g, A)}^\perp$.

Putting $k=0$ into $(\ref{dual:1})$, we easily derive that $A_0^\perp=1$.

Similarly, puting $k=1$ into $(\ref{dual:1})$, we have
$$2^m-1+A_{w_1}(2^m-1-2w_1)+A_{w_2}(2^m-1-2w_2)+A_{w_3}(2^m-1-2w_3)=0.$$
As a result, $A_1^\perp=0$.

Plugging $k=2$ and  $A_{w_1},A_{w_2},A_{w_3}$ in  $(\ref{dual:1})$, we can verify
 \begin{eqnarray*} & &\binom{2^m-1}{2}+ A_{w_1}\left[\binom{2^m-1-w_1}{2}-w_1(2^m-1-w_1)+\frac{1}{2}w_1(w_1-1)\right]\\
& &+A_{w_2}\left[\binom{2^m-1-w_2}{2}-w_2(2^m-1-w_2)+\frac{1}{2}w_2(w_2-1)\right]\\
& &+A_{w_3}\left[\binom{2^m-1-w_3}{2}-w_3(2^m-1-w_3)+\frac{1}{2}w_3(w_3-1)\right]\\
& &\ =0.
\end{eqnarray*}
 $A_2^\perp=0$ then follows.

 As for $A_3^\perp$, Equation $(\ref{dual:1})$ becomes
 \begin{eqnarray*} & &\binom{2^m-1}{3}\\
 & & +A_{w_1}\left[\binom{2^m-1-w_1}{3}-w_1\binom{2^m-1-w_1}{2}+\binom{w_1}{2}\binom{2^m-1-w_1}{1}-\binom{w_1}{3}\right]\\
& & +A_{w_2}\left[\binom{2^m-1-w_2}{3}-w_2\binom{2^m-1-w_2}{2}+\binom{w_2}{2}\binom{2^m-1-w_2}{1}-\binom{w_2}{3}\right]\\
& & +A_{w_3}\left[\binom{2^m-1-w_3}{3}-w_3\binom{2^m-1-w_3}{2}+\binom{w_3}{2}\binom{2^m-1-w_3}{1}-\binom{w_3}{3}\right]\\
& &\ = \frac{1}{3}\cdot 2^{m-1}\cdot (2^m-2^r)\cdot(2^m-2).
\end{eqnarray*}

Obviously, we have $A_3^\perp=0$ if and only if  $r=m$.

As for $A_4^\perp$, Equation $(\ref{dual:1})$ changes to
\begin{eqnarray*} & &\binom{2^m-1}{4}+ A_{w_1}\bigg[\binom{2^m-1-w_1}{4}-w_1\binom{2^m-1-w_1}{3}\\
&&+\binom{w_1}{2}\binom{2^m-1-w_1}{2}-\binom{w_1}{3}\binom{2^m-1-w_1}{1}+\binom{w_1}{4}\bigg]\\
&&+A_{w_2}\bigg[\binom{2^m-1-w_2}{4}-w_2\binom{2^m-1-w_2}{3}\\
&& +\binom{w_2}{2}\binom{2^m-1-w_2}{2}-\binom{w_2}{3}\binom{2^m-1-w_2}{1}+\binom{w_2}{4}\bigg]\\
&&+A_{w_3}\bigg[\binom{2^m-1-w_3}{4}-w_3\binom{2^m-1-w_3}{3}\\
&& +\binom{w_3}{2}\binom{2^m-1-w_3}{2}-\binom{w_3}{3}\binom{2^m-1-w_3}{1}+\binom{w_3}{4}\bigg]\\
& &= \frac{1}{3}\cdot2^{m-3}(2^m-2^r)(8-3\cdot 2^{m+1}+4^m).
\end{eqnarray*}

It can be verified that $A_4^\perp=0$ if and only if  $r=m$.

 Similarly, we can verify $A_5^\perp\neq 0$ no matter $r=m$ or not. Therefore, the dual distance of $\mathcal{C}_{(g, A)}$ is $5$ if $r=m$ and $3$ otherwise.

\end{proof}

\subsection{Parameters of $\overline{\mathcal{C}_{(g,A)}^\perp}^\perp$}
In this section, we discuss the parameters of the code  $\overline{\mathcal{C}_{(g,A)}^\perp}^\perp$. Where $\overline{\mathcal{C}_{(g,A)}^\perp}$ is the extended code of the dual of $\mathcal{C}_{(g,A)}$.

Firstly, we introduce a lemma that will be used in the upcoming theorem.

\begin{lemma}\label{thm:3}$(\cite{Dingbook}, p.70)$
Let $\mathcal{C}$ be an $[n,k,d]$ binary linear code, and let $\mathcal{C}^\perp$ be the dual of $\mathcal{C}$. Denote by $\overline{\mathcal{C}^\perp}$ the extended code of $\mathcal{C}^\perp$, and let  $\overline{\mathcal{C}^\perp}^\perp$ denote the dual of $\overline{\mathcal{C}^\perp}$. Then we have the following conclusion.

$(1)$ $\mathcal{C}^\perp$ has parameters $[n,n-k,d^\perp]$, where $d^\perp$ denotes the minimum distance of $\mathcal{C}^\perp$.

$(2)$ $\overline{\mathcal{C}^\perp}$ has parameters $[n,n-k,\overline{d^\perp}]$, where $\overline{d^\perp}$ denotes the minimum distance of $\overline{\mathcal{C}^\perp}$, which is given by
\begin{eqnarray}
     \overline{d^\perp}  &= & \left\{\begin{array}{ll}
d^\perp & {\rm if}\ d^\perp {\rm \ is\  even},\\
d^\perp+1 & {\rm if}\ d^\perp {\rm \ is\  odd}.
\end{array}
\right.
\end{eqnarray}

$(3)$ $\overline{\mathcal{C}^\perp}^\perp$ has parameters $[n+1,k+1,\overline{d^\perp}^\perp]$, where $\overline{d^\perp}^\perp$ denotes the minimum distance of $\overline{\mathcal{C}^\perp}^\perp$. Furthermore,  $\overline{\mathcal{C}^\perp}^\perp$ has only even weight codewords, and all the nonzero weights in  $\overline{\mathcal{C}^\perp}^\perp$ are the following:
$$w_1,w_2,\cdots,w-t;n+1-w_1,n+1-w_2,\cdots, n+1-w_t; n+1,$$
where $w_1, w_2, \cdots, w_t$ denote all the nonzero weights of $\mathcal{C}$.

\end{lemma}

\begin{theorem}\label{thm:4}
Let $m\geq 3$ be an odd integer. Let $\mathcal{C}_{(g,A)}$ be a binary code with parameters $[2^m-1, m+r,2^{m-1}- 2^{(m-1)/2}]$ and weight distribution in Table $\ref{table:1}$.  Then $\overline{\mathcal{C}_{(g,A)}^\perp}^\perp$ has parameters $[2^m, m+r+1, 2^{m-1}- 2^{(m-1)/2}]$$(1\leq r \leq m)$ with the following weight distribution.

\begin{table}[h]

\centering
\caption{Weight distribution of $\overline{\mathcal{C}_{(g,A)}^\perp}^\perp$}

\begin{tabular}{ |l|l| }

\hline

\hline   Weight $w$ & Multiplicity $A_w$ \\

\hline

$0$ & $1$ \\

$2^{m-1}- 2^{(m-1)/2}$ & $2^{m-1}(2^r-1)$ \\

$2^{m-1}$ & $2^{m+r}+2^m-2$ \\

$2^{m-1}+ 2^{(m-1)/2}$ & $2^{m-1}(2^r-1)$  \\

$2^m$ & $1$ \\
\hline
\end{tabular}
\label{table:2}
\end{table}
Furthermore, the dual distance of $\overline{\mathcal{C}_{(g,A)}^\perp}^\perp$ is
$6$ if $r=m$, and $4$ otherwise.

\end{theorem}

\begin{proof}
The length and dimension of $\overline{\mathcal{C}_{(g,A)}^\perp}^\perp$ can be deduced easily from Lemma $\ref{thm:3}$. Next, we discuss its weight distribution. As we know, the possible weights of  $\overline{\mathcal{C}_{(g,A)}^\perp}^\perp$  are $0,\  2^{m-1}- 2^{(m-1)/2}, \ 2^{m-1}, \ 2^{m-1}+ 2^{(m-1)/2},\  2^m$. But the weights $0$ and $w_4=2^m$ occur only once. So  we need to decide the multiplicity of weights $w_1=2^{m-1}- 2^{(m-1)/2}$, $w_2=2^{m-1}$, $w_3=2^{m-1}+ 2^{(m-1)/2}$.

The first three Pless power moments lead to
\begin{equation}
   \left\{
   \begin{aligned}
1+A_{w_1}+A_{w_2}+A_{w_3}+1&=2^{m+r+1},\\
w_1A_{w_1}+w_2A_{w_2}+w_3A_{w_3}+w_4&=2^{m+r}\cdot2^m,\\
w_1^2A_{w_1}+w_2^2A_{w_2}+w_3^2A_{w_3}+w_4^2&=2^{m+r-1}\cdot(2^m+1)\cdot2^m.
\end{aligned}
\right.
\end{equation}

Solving the above linear equations yields the following result:
 \begin{equation}
   \left\{
   \begin{aligned}
A_{w_1}&=2^{m-1}(2^r-1),\\
A_{w_2}&=2^{m+r}+2^m-2,\\
A_{w_3}&=2^{m-1}(2^r-1).
\end{aligned}
\right.
\end{equation}

Combing the conclusion of  Lemma $\ref{thm:3}(3)$ and the dual distance of $\mathcal{C}_{(g,A)}$ in Theorem $\ref{thm:2}$, we can easily deduce that the dual distance of $\overline{\mathcal{C}_{(g,A)}^\perp}^\perp$ is
$6$ if $r=m$, and $4$ otherwise.
\end{proof}

When $A=\{0\}$,  $\overline{\mathcal{C}_{(g, A)}^\perp}^\perp$ is the first-order Reed-Muller code. When $|A|\geq 2$, the  distance of $\overline{\mathcal{C}_{(g, A)}^\perp}$  is at leat $4$.

Actually, we can obtain a large family of optimal codes from the constructed codes $\mathcal{C}_{(g, A)}$, and their related codes $\mathcal{C}_{(g, A)}^\perp$, $\overline{\mathcal{C}_{(g, A)}^\perp}$, and $\overline{\mathcal{C}_{(g, A)}^\perp}^\perp$.

\begin{example}\label{exam:1}
If $m=5$, then linear codes with parameters $[31, 8, 12]$, $[31, 9, 12]$, $[31, 10, 12]$ provided by $\mathcal{C}_{(g, A)}$,  $[31, 21, 5]$ provided by $\mathcal{C}_{(g, A)}^\perp$, $[32, 21, 6]$, $[32, 23, 4]$, $[32, 24, 4]$ provided by $\overline{\mathcal{C}_{(g, A)}^\perp}$, and $[32, 9, 12]$, $[32, 10, 12]$, $[32, 11, 12]$ provided by $\overline{\mathcal{C}_{(g, A)}^\perp}^\perp$ are all optimal binary linear codes according to
the Code Table at http://www.codetables.de/.

There are also many almost optimal codes that we will not list here, the reader can check them by changing their parameters according to the preceding theorems.
\end{example}

\subsection{Designs from $\overline{\mathcal{C}_{(g,A)}^\perp}^\perp$}

Let $\mathcal{P}$ be a set of $n\geq 1$ elements, and let $\mathcal{B}$ be a set of $k-$subsets of $\mathcal{P}$, where $k$ is a positive integer with $1\leq k \leq n$. Let $t$ be a positive integer with $t\leq k$. The pair $\mathbb{D}=(\mathcal{P}, \mathcal{B})$ is called a \emph{$t-(n, k, \lambda)$ design}, or simply \emph{$t-$design}, if every $t-$subset of $\mathcal{P}$ is contained in exactly $\lambda$ elements of $\mathcal{B}$. The elements of $\mathcal{P}$ are called points, and those of $\mathcal{B}$ are referred to as blocks.

A necessary condition for the existence of a $t-(n, k, \lambda)$ design is that
$$\binom{k-i}{t-i}\ \  \text{divides}\ \  \lambda\binom{n-i}{t-i},$$
for all integers $i$ with $0\leq i \leq t.$

Let $\mathcal{C}$ be an $[n, k, d]$ linear code over $\text{GF}(p^m)$. let $A_i$ denote the number of codewords with Hamming weight $i$ in $\mathcal{C}$, where
$0\leq i \leq n$. For each $k$ with $A_k\neq 0$, let $\mathcal{B}_k$ denote the set of the supports of all codewords with Hamming weight $k$ in $\mathcal{C}$, where the coordinates of a codeword are indexed by $(0,1,2,\cdots,n-1)$. Let $\mathcal{P}=\{0, 1, 2, \cdots, n-1\}$. Assmus and Mattson showed that the pair $(\mathcal{P}, \mathcal{B}_k)$ may be a $t-(n, k, \lambda)$ design under certain conditions\cite{AssmusM74},\cite{Huffman03}(p. 303).

\begin{theorem}\label{thm:Ass}
$(\text{Assmus-Mattson Theorem})$ Let $\mathcal{C}$ be a binary $[n, k, d]$ code. Suppose $\mathcal{C}^\perp$ has minimum weight $d^\perp$,  $A_i=A_i(\mathcal{C})$ and $A_i^\perp=A_i(\mathcal{C}^\perp)$ $(0\leq i \leq n)$ are the weight distribution of $\mathcal{C}$ and $\mathcal{C}^\perp$ respectively. Fix a positive integer $t$ with $t<d$, and let $s$ be the number of $i$ with $A_i^\perp \neq 0$ for $0< i \leq n-t$. Suppose that $s\leq d-t$. Then

$\bullet$ the codewords of weight $i$ in $\mathcal{C}$ hold a $t-$design provided that $A_i\neq 0$ and $d\leq i\leq n$,

$\bullet$ the codewords of weight $i$ in $\mathcal{C}^\perp$ hold a $t-$design provided that $A_i^\perp\neq 0$ and $d^\perp\leq i\leq n-t$.
\end{theorem}

\begin{theorem}\label{thm:5}
let $m\geq 3$ be an odd integer.  Let  $\overline{\mathcal{C}_{(g, A)}^\perp}^\perp$ be the code constructed in Theorem $\ref{thm:4}$ from almost bent function $g$. Let $\mathcal{P}=\{0, 1, 2, \cdots, 2^m-1\}$, and let $\mathcal{B}$ be the set of the supports of the codewords of $\overline{\mathcal{C}_{(g, A)}^\perp}^\perp$ with weight $k$, where $A_k\neq 0$ is given in Table  $\ref{table:2}$.

If $r=m$, then $(\mathcal{P}, \mathcal{B})$ is a $3-(2^m, k, \lambda)$ design with
$$\lambda=\frac{k(k-1)(k-2)A_k}{2^m(2^m-1)(2^m-2)}.$$

If $r\neq m$, then $(\mathcal{P}, \mathcal{B})$ is a $1-(2^m, k, \lambda)$ design with
$$\lambda=\frac{kA_k}{2^m}.$$

\end{theorem}
\begin{proof}
The weight distribution of $\overline{\mathcal{C}_{(g,A)}^\perp}^\perp$ is given  in Table $\ref{table:2}$ of Theorem $\ref{thm:4}$. The minimum distance of $\overline{\mathcal{C}_{(g, A)}^\perp}^\perp$ is equal to $6$ if $r=m$, and $4$ otherwise. Put $t=3$($t=1$ for the case of $r\neq m$), the number of $i$ with $A_i\neq 0$ and $1\leq i \leq 2^m-1-t$ is $s=3$. Hence, $s=d^\perp -t$. The desired conclusion then follows from Theorem $\ref{thm:Ass}$ and the fact that two binary vectors have the same support if and only if they are equal.
\end{proof}

\begin{example}
Let $m\geq 3$ be an odd integer and let $\overline{\mathcal{C}_{(g, A)}^\perp}^\perp$ be a binary code with parameters $[2^m, m+r+1, 2^{m-1}-2^{(m-1)/2}]$ and weight distribution in Table $\ref{table:2}$.

 If $r=m$, then the code  $\overline{\mathcal{C}_{(g, A)}^\perp}^\perp$   holds three $3-$designs with the following parameters:

$\bullet (n,k,\lambda)=\left(2^m, 2^{m-1}- 2^{(m-1)/2}, \frac{(2^{m-2}-2^{\frac{m-3}{2}})(2^{m-1}-2^{\frac{m-1}{2}}-1)(2^{m-2}-2^{\frac{m-3}{2}}-1)}{2^{m-1}-1}\right)$.

$\bullet (n, k, \lambda)=\left(2^m, 2^{m-1}, \frac{(2^{2m-1}+2^{m-1}-1)(2^{m-2}-1)}{2^{m}-1}\right)$.

$\bullet (n, k, \lambda)=\left(2^m, 2^{m-1}+ 2^{(m-1)/2}, \frac{(2^{m-2}+2^{\frac{m-3}{2}})(2^{m-1}+2^{\frac{m-1}{2}}-1)(2^{m-2}+2^{\frac{m-3}{2}}-1)}{2^{m-1}-1}\right)$.

If $r\neq m$, then the code  $\overline{\mathcal{C}_{(g,A)}^\perp}^\perp$   holds three families of $1-$designs with the following parameters:

$\bullet (n, k, \lambda)=\left(2^m, 2^{m-1}- 2^{(m-1)/2}, (2^r-1)(2^{m-2}-2^{\frac{m-3}{2}})\right)$.

$\bullet (n, k, \lambda)=\left(2^m, 2^{m-1}, 2^{m+r-1}+2^{m-1}-1\right)$.

$\bullet (n, k, \lambda)=\left(2^m, 2^{m-1}+ 2^{(m-1)/2}, (2^r-1)(2^{m-2}+ 2^{\frac{m-3}{2}})\right)$.

\end{example}

\section{Linear codes from planar functions}

In this section, we discuss  a family of linear codes from planar functions, analyze their parameters as well as the parameters of  their dual codes and extended codes.

Throughout this section, we assume $p$ is an odd prime number, $m$ is a positive integer, and $f(x)$ is a planar function from $\text{GF}(p^m)$ to itself with $f(0)=0$.

Define linear codes from $f(x)$ as
\begin{equation}\label{def:2}
\mathcal{C}_{f,A}=\{\left(\mbox{Tr}_1^m(af(x)+bx)\right)_{x\in GF(p^m)^*}: a\in A,b\in \mbox{GF}(p^m)\},
\end{equation}
where $A$ is an additive subgroup of $(\text{GF}(p^m),+)$ with order $p^r$,  $0\leq r \leq m$.

For any fixed $a\in A, b\in \text{GF}(p^m)$, let $c_{a,b}$ denote the corresponding codeword defined by
$$c_{a,b}=(\text{Tr}_1^m(af(x)+bx))_{x\in \text{GF}(p^m)^*}.$$

All the known planar functions from $\text{GF}(p^m)$ to itself are listed in Section $\ref{sec:planar}$. Subsequently, we will discuss individually the parameters of linear codes constructed from these  planar functions.

\subsection{Linear codes from $f_1(x)=x^{p^t+1}$}
\begin{theorem}\label{thm:8}
Let $p$ be an odd prime number, $m\geq 3$ be an odd integer. Let $f_1(x)=x^{p^t+1}$ be a planar function from $\text{GF}(p^m)$ to $\text{GF}(p^m)$, where $t\geq 0$ is an integer and $m/gcd(m, t)$ is odd. Then the constructed code $\mathcal{C}_{f_1, A}$ in $(\ref{def:2})$ has  parameters $[p^m-1, m+r, (p-1)p^{m-1}- p^{\frac{m-1}{2}}]$ with weight distribution in Table $\ref{table:3}.$

\begin{table}[h]

\centering
\caption{Weight distribution of $\mathcal{C}_{f_1,A}$}

\begin{tabular}{ |l|l| }

\hline

\hline   Weight $w$ & Multiplicity $A_w$ \\

\hline

$0$ & $1$ \\

$(p-1)p^{m-1}- p^{\frac{m-1}{2}}$ &$\frac{p-1}{2}\left[{p^{\frac{m-1}{2}+r}+p^{r-1}(-2+p+p^m)-p^\frac{m-1}{2}}-(p-1)p^{m-1}\right]$\\

$(p-1)p^{m-1}$ &  $p^{m+r-1}+p^{m+1}-2p^m+p^{m-1}-p^{r+1}+3p^r-2p^{r-1}-1 $\\

$(p-1)p^{m-1}+ p^{\frac{m-1}{2}}$ & $\frac{p-1}{2}\left[p^{r-1}(p^m+p-2)-p^{\frac{m-1}{2}+r}+p^{\frac{m-1}{2}}-(p-1)p^{m-1}\right]$ \\

\hline
\end{tabular}
\label{table:3}
\end{table}

Moreover,  its dual code has parameters $[p^{m}-1, p^{m}-1-m-r, d^\perp]$, where $d^\perp\geq3.$  Especially, if $p=3$ and $r=m,$ we have $d^\perp=4$.
\end{theorem}
\begin{proof}
It can be deduced that the length of the code $\mathcal{C}_{f_1,A}$ is $n=p^m-1$, and the dimension is $k=m+r$.  We can also prove that $A_1^\perp, A_2^\perp=0$ similar to the proof of Theorem $\ref{thm:1}$. Authors in \cite{DingYuan06} and \cite{Fengluo07} proved independently that the possible weights of the constructed linear code are $w_1=(p-1)p^{m-1}- p^{\frac{m-1}{2}}$, $w_2=(p-1)p^{m-1}$, $w_3=(p-1)p^{m-1}+ p^{\frac{m-1}{2}}.$

According to the first three Pless power moments, we have
\begin{equation}
   \left\{
   \begin{aligned}
1+A_{w_1}+A_{w_2}+A_{w_3}&=p^{m+r},\\
w_1A_{w_1}+w_2A_{w_2}+w_3A_{w_3}&=p^{m+r-1}(p-1)(p^m-1),\\
w_1^2A_{w_1}+w_2^2A_{w_2}+w_3^2A_{w_3}&=p^{m+r-2}(p-1)(p^m-1)(p^{m+1}-p^m-p+2).
\end{aligned}
\right.
\end{equation}

The weight distribution of $\mathcal{C}_{f_1, A}$ then follows by solving this system of  linear equations
\begin{equation}
   \left\{
   \begin{aligned}
A_{w_1}&=\frac{p-1}{2}\left[{p^{\frac{m-1}{2}+r}+p^{r-1}(-2+p+p^m)-p^\frac{m-1}{2}}-(p-1)p^{m-1}\right],\\
A_{w_2}&=p^{m+r-1}+p^{m+1}-2p^m+p^{m-1}-p^{r+1}+3p^r-2p^{r-1}-1,\\
A_{w_3}&=\frac{p-1}{2}\left[p^{r-1}(p^m+p-2)-p^{\frac{m-1}{2}+r}+p^{\frac{m-1}{2}}-(p-1)p^{m-1}\right].
\end{aligned}
\right.
\end{equation}

Next, we consider the weight distribution and distance of the dual code of   $\mathcal{C}_{f_1, A}$.

It is easy to obtain the weight enumerator of  $\mathcal{C}_{f_1, A}$
\begin{equation}
A(z)=\sum_{i=0}^nA_iz^i=1+A_{w_1}z^{w_1}+A_{w_2}z^{w_2}+A_{w_3}z^{w_3}.
\end{equation}

According to the MacWilliams Identity, we have
 \begin{eqnarray*}
 p^{m+r}A^\perp(z)&=&(1+(p-1)z)^{p^m-1}\cdot A(\frac{1-z}{1+(p-1)z})\\
&=&(1+(p-1)z)^{p^m-1}\cdot \bigg[1+A_{w_1}\cdot \frac{(1-z)^{w_1}}{(1+(p-1)z)^{w_1}}\\
& & +A_{w_2}\cdot \frac{(1-z)^{w_2}}{(1+(p-1)z)^{w_2}}+A_{w_3}\cdot \frac{(1-z)^{w_3}}{(1+(p-1)z)^{w_3}}\bigg]\\
&=&(1+(p-1)z)^{p^m-1}+A_{w_1}\cdot (1-z)^{w_1}(1+(p-1)z)^{p^m-1-w_1}\\
&&+A_{w_2}\cdot (1-z)^{w_2}(1+(p-1)z)^{p^m-1-w_2}+A_{w_3}\cdot (1-z)^{w_3}(1+(p-1)z)^{p^m-1-w_3}.
\end{eqnarray*}

Denote

 \begin{equation}
   \left\{
   \begin{aligned}
U_1(k)&=\sum_{i+j=k}(-1)^i\binom{w_1}{i}\binom{p^m-1-w_1}{j}(p-1)^j,\\
U_2(k)&=\sum_{i+j=k}(-1)^i\binom{w_2}{i}\binom{p^m-1-w_2}{j}(p-1)^j,\\
U_3(k)&=\sum_{i+j=k}(-1)^i\binom{w_3}{i}\binom{p^m-1-w_3}{j}(p-1)^j.
\end{aligned}
\right.
\end{equation}
Then it can be proved that the value of $p^{m+r}A_k^\perp$ is exactly the coefficient of $z^k$ on the right hand side of the above equation.
That is,
\begin{equation} \label{mac:22}
\binom{p^m-1}{k}(p-1)^k+A_{w_1}U_1(k)+A_{w_2}U_2(k)+A_{w_3}U_3(k).
\end{equation}

Putting $k=0$ into $(\ref{mac:22})$, we easily derive that $A_0^\perp=1$.

Similarly,  putting $k=1$ into $(\ref{mac:22})$, we obtain
 \begin{eqnarray*} & &(p^m-1)(p-1)+A_{w_1}\left[(p^m-1-w_1)(p-1)-w_1\right]+A_{w_2}\left[(p^m-1-w_2)(p-1)-w_2\right]\\
&+& A_{w_3}\left[(p^m-1-w_3)(p-1)-w_3\right]=0.
\end{eqnarray*}
It then follows that $A_1^\perp=0$.

Plugging $k=2$ and $A_{w_1}, A_{w_2}, A_{w_3}$ into $(\ref{mac:22})$ leads to
 \begin{eqnarray*} & &\binom{p^m-1}{2}(p-1)^2\\
 & &+ A_{w_1}\left[\binom{w_1}{0}\binom{p^m-1-w_1}{2}(p-1)^2-\binom{w_1}{1}\binom{p^m-1-w_1}{1}(p-1)+\binom{w_1}{2}\binom{p^m-1-w_1}{0}\right]\\
& &+A_{w_2}\left[\binom{w_2}{0}\binom{p^m-1-w_2}{2}(p-1)^2-\binom{w_2}{1}\binom{p^m-1-w_2}{1}(p-1)+\binom{w_2}{2}\binom{p^m-1-w_2}{0}\right]\\
& &+A_{w_3}\left[\binom{w_3}{0}\binom{p^m-1-w_3}{2}(p-1)^2-\binom{w_3}{1}\binom{p^m-1-w_3}{1}(p-1)+\binom{w_3}{2}\binom{p^m-1-w_3}{0}\right]\\
&&=0.
\end{eqnarray*}
As a result,  $A_2^\perp=0$.

As for $A_3^\perp$, we have
 \begin{eqnarray*} & &p^{m+r}A_3^\perp=\binom{p^m-1}{3}(p-1)^3\\
 & &\ + A_{w_1}\cdot\bigg[\binom{w_1}{0}\binom{p^m-1-w_1}{3}(p-1)^3-\binom{w_1}{1}\binom{p^m-1-w_1}{2}(p-1)^2\\
 & &\ +\binom{w_1}{2}\binom{p^m-1-w_1}{1}(p-1)-\binom{w_1}{3}\binom{p^m-1-w_1}{0}\bigg]\\
& &\  +A_{w_2}\cdot\bigg[\binom{w_2}{0}\binom{p^m-1-w_2}{3}(p-1)^3-\binom{w_2}{1}\binom{p^m-1-w_2}{2}(p-1)^2\\
&&\ +\binom{w_2}{2}\binom{p^m-1-w_2}{1}(p-1)-\binom{w_2}{3}\binom{p^m-1-w_2}{0}\bigg]\\
& &\ +A_{w_3}\cdot\bigg[\binom{w_3}{0}\binom{p^m-1-w_3}{3}(p-1)^3-\binom{w_3}{1}\binom{p^m-1-w_3}{2}(p-1)^2\\
& &\ +\binom{w_3}{2}\binom{p^m-1-w_3}{1}(p-1)-\binom{w_1}{3}\binom{p^m-1-w_3}{0}\bigg]\\
& &\ =\frac{1}{6}(p-1)\cdot p^{m}\cdot \left[(p-1)^2p^{2m}-p^r(6+(p-6)p)-p^m(p+p^r(3p-5))\right].
\end{eqnarray*}
It can be proved that $A_3^\perp=0$ if  $m=r$ and $p=3$.

As for $A_4^\perp$,  $(\ref{mac:22})$ turns to
 \begin{eqnarray*} &&p^{m+r}A_4^\perp= \binom{p^m-1}{4}(p-1)^4\\
 & &\ \ + A_{w_1}\cdot\bigg[\binom{w_1}{0}\binom{p^m-1-w_1}{4}(p-1)^4-\binom{w_1}{1}\binom{p^m-1-w_1}{3}(p-1)^3\\
 & &\ \ +\binom{w_1}{2}\binom{p^m-1-w_1}{2}(p-1)^2-\binom{w_1}{3}\binom{p^m-1-w_1}{1}(p-1)+\binom{w_1}{4}\binom{p^m-1-w_1}{0}\bigg]\\
 & &\ \ + A_{w_2}\cdot\bigg[\binom{w_2}{0}\binom{p^m-1-w_2}{4}(p-1)^4-\binom{w_2}{1}\binom{p^m-1-w_2}{3}(p-1)^3\\
 & &\ \ +\binom{w_2}{2}\binom{p^m-1-w_2}{2}(p-1)^2-\binom{w_2}{3}\binom{p^m-1-w_2}{1}(p-1)+\binom{w_2}{4}\binom{p^m-1-w_2}{0}\bigg]\\
 & &\ \ + A_{w_3}\cdot\bigg[\binom{w_3}{0}\binom{p^m-1-w_3}{4}(p-1)^4-\binom{w_3}{1}\binom{p^m-1-w_3}{3}(p-1)^3\\
 & &\ \ +\binom{w_3}{2}\binom{p^m-1-w_3}{2}(p-1)^2-\binom{w_3}{3}\binom{p^m-1-w_3}{1}(p-1)+\binom{w_3}{4}\binom{p^m-1-w_3}{0}\bigg]\\
& &\ \ =\frac{1}{24}(p-1)\cdot p^{m}\cdot [p^m(16-10p+(p-1)p^m((p-1)^2p^m-p))\\
&&\ \ \ +p^r(-72+(3-2p)p^{2m}+2p^m(27+2p(6p-16))+p(120+p(9p-62)))].
\end{eqnarray*}

If $p=3$ and $r=m$, we get $p^{m+r}A_4^\perp=\frac{5}{4}\cdot 3^{2m-1}(9^m-4\times 3^m+3)\neq 0.$

The dual distance $d^\perp\geq3$  then follows. Especially, if $p=3$ and $r=m$, we have $d^\perp=4$.

\end{proof}

\subsection{Linear codes from $f_2(x)=x^{\frac{3^k+1}{2}}$ and $f_3(x)=x^{10}-ux^6-u^2x^2$}
If $f(x)$ is a planar function with form $f_2(x)=x^{\frac{3^k+1}{2}}$ or $f_3(x)=x^{10}-ux^6-u^2x^2$. It was proved in \cite{DingYuan06} that when $p=3$, $m$ is odd, the possible nonzero Hamming weights of codewords in $\mathcal{C}_{f,A}$ are $2\cdot 3^{m-1}$ and $2\cdot 3^{m-1}\pm 3^{\frac{m-1}{2}}$. Then we have the following conclusion.

 \begin{theorem}\label{thm:6}
 Let $p=3$. Let $m\geq 3$ be an odd integer. If $f(x)=x^{\frac{3^k+1}{2}}$ or $f(x)=x^{10}-ux^6-u^2x^2$, then the code $\mathcal{C}_{f, A}$ constructed in $(\ref{def:2})$  from $f(x)$ has parameters $[3^m-1, m+r, 2\cdot 3^{m - 1} - 3^{(m - 1)/2}]$ with  the  weight distribution in Table $\ref{table:4}$.

\begin{table}[h]

\centering
\caption{Weight distribution of codes from $\mathcal{C}_{f,A}$}

\begin{tabular}{ |l|l| }

\hline

\hline   Weight $w$ & Multiplicity $A_w$ \\
\hline

$0$ & $1$ \\

$2\cdot 3^{m - 1} - 3^{(m - 1)/2}$ & $3^{m + r - 1}-3^{(m - 1)/2} - 2\cdot 3^{m - 1} + 3^{r - 1} + 3^{(m - 1)/2 + r}
 $ \\

$2\cdot 3^{m - 1}$ & $3^{m + r - 1} - 2\cdot 3^{r - 1} + 4\cdot 3^{m - 1} - 1$ \\

$2\cdot 3^{m - 1} + 3^{(m - 1)/2}$ & $3^{(m - 1)/2} - 2\cdot 3^{m - 1} + 3^{r - 1} -3^{(m - 1)/2 + r} +
 3^{m + r - 1}$ \\

\hline
\end{tabular}
\label{table:4}
\end{table}

Furthermore, the dual distance $d^\perp$ of  $\mathcal{C}_{f, A}$ is $4$ if $r=m$ and $3$ otherwise.
 \end{theorem}

 \begin{proof}
One can easily see that the length of the code $\mathcal{C}_{f, A}$ is $n=3^m-1$, the dimension is $k=m+r$, and $A_1^\perp, A_2^\perp=0$. According to the first three Pless power moments, we have
 \begin{equation}
   \left\{
   \begin{aligned}
1+A_{w_1}+A_{w_2}+A_{w_3}&=3^{m+r},\\
w_1A_{w_1}+w_2A_{w_2}+w_3A_{w_3}&=3^{m+r-1} \cdot 2(3^m-1),\\
w_1^2A_{w_1}+w_2^2A_{w_2}+w_3^2A_{w_3}&=3^{m+r-2}\cdot (2-2\cdot3^{1+m}+ 4\cdot 9^m),\\
\end{aligned}
\right.
\end{equation}
where $w_1=2\cdot 3^{m-1}- 3^{\frac{m-1}{2}}$, $w_2=2\cdot 3^{m-1}$,  $w_3=2\cdot 3^{m-1}+ 3^{\frac{m-1}{2}}.$

The weight distribution of $\mathcal{C}_{f, A}$ then follows by solving the above linear equations. Specifically,
\begin{equation}
   \left\{
   \begin{aligned}
A_{w_1}&=-3^{(m - 1)/2} - 2\cdot 3^{m - 1} + 3^{r - 1} + 3^{(m - 1)/2 + r} +
 3^{m + r - 1},\\
A_{w_2}&=3^{m + r - 1} - 2\cdot 3^{r - 1} + 4\cdot 3^{m - 1} - 1,\\
A_{w_3}&=3^{(m - 1)/2} - 2\cdot 3^{m - 1} + 3^{r - 1} -3^{(m - 1)/2 + r} +
 3^{m + r - 1}.
\end{aligned}
\right.
\end{equation}

Consequently, the weight enumerator of  $\mathcal{C}_{f, A}$ can be expressed as:

\begin{equation}
A(z)=\sum_{i=0}^nA_iz^i=1+A_{w_1}z^{w_1}+A_{w_2}z^{w_2}+A_{w_3}z^{w_3}.
\end{equation}

The MacWilliams Identity leads to

 \begin{eqnarray*}3^{m+r}A^\perp(z) &= &(1+2z)^{n}\cdot A(\frac{1-z}{1+2z})\\
 &=& (1+2z)^{n}\cdot \bigg[1+A_{w_1}\cdot \frac{(1-z)^{w_1}}{(1+2z)^{w_1}}\\
 & & +A_{w_2}\cdot \frac{(1-z)^{w_2}}{(1+2z)^{w_2}}+A_{w_3}\cdot \frac{(1-z)^{w_3}}{(1+2z)^{w_3}}\bigg]\\
&=&(1+2z)^{n}+A_{w_1}\cdot (1-z)^{w_1}(1+2z)^{n-w_1}\\
& &+A_{w_2}\cdot (1-z)^{w_2}(1+2z)^{n-w_2}+A_{w_3}\cdot (1-z)^{w_3}(1+2z)^{n-w_3},
\end{eqnarray*}
where $n=3^m-1$.

Denote

 \begin{equation}
   \left\{
   \begin{aligned}
U_1(k)&=\sum_{i+j=k}(-1)^i\binom{w_1}{i}\binom{3^m-1-w_1}{j}\cdot 2^j,\\
U_2(k)&=\sum_{i+j=k}(-1)^i\binom{w_2}{i}\binom{3^m-1-w_2}{j}\cdot 2^j,\\
U_3(k)&=\sum_{i+j=k}(-1)^i\binom{w_3}{i}\binom{3^m-1-w_3}{j}\cdot 2^j.
\end{aligned}
\right.
\end{equation}
Then the value of $3^{m+r}A_k^\perp$ is exactly the coefficient of $z^k$ on the right hand side of the above equation.
That is
\begin{equation}\label{mac:2}
\binom{n}{k}2^k+A_{w_1}U_1(k)+A_{w_2}U_2(k)+A_{w_3}U_3(k).
\end{equation}

Put $k=1, 2, 3$ individually into Equation $(\ref{mac:2})$, we have $A_1^\perp=A_2^\perp=0 $,  and
$$3^{m+r}A_3^\perp=3^{m-1}\cdot (4\cdot 3^m-3)\cdot (3^m-3^r).$$
It is obvious that $A_3^\perp=0$ when $r=m$.

Similarly, we have
$$3^{m+r}A_4^\perp=\frac{1}{4}\cdot 3^{m-1}\cdot (14\cdot 3^{m+1}-3^{3+r}+14\cdot 3^{m+r+1}-3^{2m+r+1}-62\cdot 9^m+8\cdot 27^m).$$

It can be verified that $A_4^\perp \neq 0$  even if $r=m.$
So the dual distance $d^\perp$ of  $\mathcal{C}_{f, A}$ is $4$ if $r=m$ and $3$ otherwise.
\end{proof}

\begin{theorem}\label{thm:7}
Let $p=3.$ Let $\mathcal{C}_{f,A}$  be the code presented in Theorem $\ref{thm:6}$. Then the code $\overline{\mathcal{C}_{(f,A)}^\perp}^\perp$ has parameters $[3^m,m+r+1, 2\cdot 3^{m - 1} - 3^{(m - 1)/2}]$ with weight distribution in Table $\ref{table:5}$.
\begin{table}[h]

\centering
\caption{Weight distribution of $\overline{\mathcal{C}_{(f,A)}^\perp}^\perp$}

\begin{tabular}{ |l|l| }

\hline

\hline   Weight $w$ & Multiplicity $A_w$ \\
\hline

$0$ & $1$ \\

$2\cdot 3^{m - 1} - 3^{(m - 1)/2}$ & $ 3^{m }  (3^r-1) $ \\

$2\cdot 3^{m - 1}$ & $3^{m + r } + 2\cdot 3^{m } - 3$ \\

$2\cdot 3^{m - 1} + 3^{(m - 1)/2}$ & $3^{m}  (3^r-1) $ \\

$3^m$ & $2$\\

\hline
\end{tabular}
\label{table:5}
\end{table}

\end{theorem}
\begin{rmk}
Here we omit the proof of this theorem because its proof is similar to the proof of the previous theorems. We just leave the following values about the dual of the constructed codes for reference:

$A_1^\perp=A_2^\perp=0$, $A_3^\perp=9^m(3^m-3^r)$, $A_4^\perp=\frac{1}{4}\cdot 3^{2m+1}(3^m-3)(3^m-3^r)$,
$A_5^\perp=\frac{1}{4}\cdot 9^m(3^m-3)(-7\times 3^m+2\times 3^{r+1}+9^m)$.

Hence, the minimum distance  of $\overline{\mathcal{C}_{(f,A)}^\perp}$ is $5$ if $r=m$, and $3$ otherwise.
\end{rmk}

If $p=3$, all the codes defined from the known planar functions in $(\ref{def:2})$ have the same parameters and the same weight distribution. We can also obtain many optimal codes  from the codes that we have constructed via planar functions.
\begin{example}\label{exam:2}
If $p=3$ and $m=3$, the constructed codes with parameters $[26, 6, 15]$, $[26, 5, 15]$ from $\mathcal{C}_{(f,A)}$, $[26,20,4]$, $[26,21,3]$, $[26,22,3]$ from $\mathcal{C}_{(f,A)}^\perp$, $[27,6,15]$ from $\overline{\mathcal{C}_{(f,A)}^\perp}$, $[27,21,4]$, $[27,22,3]$, $[27,23,3]$ from $\overline{\mathcal{C}_{(f,A)}^\perp}^\perp$ are all optimal  linear codes according to the Code Table at http://www.codetables.de/.

If $p=3$ and $m=5$, we also obtain optimal codes with parameters $[242,10, 153]$, $[242,9,153]$, $[242,232, 4]$, $[242, 235, 3]$, $[242, 236, 3]$, $[243, 10, 153]$, $[243, 9, 153]$, $[243, 233,4]$, $[243, 236, 3]$, $[243, 237,3]$ , and almost optimal codes with parameters $[242, 233,3]$, $[242, 234,3]$, $[243, 234, 4]$, $[243, 235, 3]$.

If $p=5$, $m=3$, we have optimal codes  with parameters$[124,6,95]$, $[124, 119,3]$, $[124,120,3]$, $[125,6,95]$, $[125,120,3]$, $[125,121,3]$, and almost optimal codes  with parameters $[124,118,3]$, $[125,119,3]$.

\end{example}

\section{Secret sharing schemes from some of the constructed codes}

In this section, we discuss the secret sharing schemes from the codes presented in this paper.

A secret sharing scheme consists of

- a dealer;

- a group $\mathcal{P}=\{P_1, P_2, \cdots, P_l\}$ of $l$ participants;

- a secret space $\mathcal{S}$ with the uniform probability distribution;

- $l$ share spaces $S_1, S_2, \cdots, S_l$;

- a share computing procedure $\mathcal{F}$; and

- a secret recovering procedure $\mathcal{G}$.

 The dealer randomly chooses an element $s\in S$ as the secret to be shared among the
participants in $\mathcal{P}$, and then employ the sharing computing procedure $\mathcal{F}$ to compute shares $(s_1, s_2, \cdots , s_l
) $  with  possibly
some other random parameters. Finally the dealer distributes $s_i$ to Participant $P_i$ as his/her share of
the secret $s$ for each $i$ with $1 \leq i \leq l$.

If a subset of participants $\{P_{i1} , P_{i2}, \cdots, P_{im}
\}(m\leq l)$ is able to recover $s$ from their shares,
we call the subset $\{P_{i1} , P_{i2}, \cdots, P_{im}
\}$ an access set.
All the access sets form a set called
the access structure. An access set is minimal
if any proper subset of it is not an access set. The set of all minimal access sets is called the
minimal access structure.

A secret sharing scheme is said to be $t$-democratic if every group of $t$ participants is in
the same number of minimal access sets, where $t\geq 1$. A participant is called a dictator if
she/he is a member of every minimal access set.

Secret sharing schemes have applications in banking system, cryptographic protocols, and the control of nuclear weapons.  Two constructions of secret sharing schemes from linear codes have been investigated in the literature. One of them is to construct secret sharing scheme from the duals of  minimal codes.

For a code $\mathcal{C}$ over GF$(q)$, the support of the codeword $c=\{c_0, c_1, \cdots, c_{n-1}\}\in \mathcal{C}$ is defined by
$$\text{supp}(c)=\{i: c_i\neq 0, 0\leq i\leq n-1\}.$$
We call a codeword $c\in \mathcal{C}$ covers a codeword $c'\in \mathcal{C}$ if $\text{supp}(c')\subseteq \text{supp}(c)$.

A codeword $c\in \mathcal{C}$  is minimal if it covers only codewords of the form $a\cdot c$, where $a\in \text{GF}(q)$. A linear code $\mathcal{C}$ is minimal if every codeword $c\in \mathcal{C}$ is minimal. There is a sufficient condition for a linear code to be minimal in \cite{Ashikhmin}:

\begin{lemma}\label{lem:2}
For an $[n,k,d]$ code $\mathcal{C}$ over  $\text{GF}(q)$, if $\frac{w_{min}}{w_{max}}>\frac{q-1}{q}$, where $w_{min}$ and $w_{max}$ denote the minimum and maximum nonzero weights of $\mathcal{C}$ respectively, then $\mathcal{C}$ is  minimal.
\end{lemma}

By this lemma, we have the following facts:

$\bullet$ Linear codes $\mathcal{C}_{g, A}$ presented in Theorem  $\ref{thm:1}$ are minimal provided that $m>3$.

$\bullet$ Linear codes $\mathcal{C}_{f, A}$ presented in Theorem $\ref{thm:8}$, $\ref{thm:6}$ are minimal.

The following theorem describes the access structure of the secret sharing
scheme based on the dual of a minimal linear code that was developed by Ding
and Yuan in \cite{Dingsecret}.
\begin{theorem}\label{thm:secret}
Let $\mathcal{C}$ be an $[n, k, d]$ code over GF($q$), and let $H = [h_0, h_1, \cdots, h_{n-1}]$
be its parity-check matrix. If $C^\perp$ is minimal, then in the secret sharing scheme
based on $C$, the set of participants is
 $\mathcal{P}= \{P_1, P_2, \cdots, P_{n-1}\}$, and there are altogether $q^{n-k-1}$ minimal access sets.

$\bullet$ When $d = 2$, the access structure is as follows.
If $h_i$ is a scalar multiple of $h_0, 1 \leq i \leq n-1$, then Participant $P_i$ must be in every minimal access set. If $h_i$ is not a scalar multiple of $h_0, 1 \leq i \leq n-1$, then Participant $P_i$ must be in $(q-1)q^{n-k-2}$ out of $q^{n-k-1}$ minimal access sets.

 $\bullet$ When $d \geq 3$, for any fixed $1 \leq t \leq min\{n-k-1, d-2\} $, every group of $t$ participants is involved in $(q -1)^tq^{n-k-(t+1)}$ out of $q^{n-k-1}$ minimal access sets.
\end{theorem}

Theorem $\ref{thm:secret}$ shows that the secret sharing scheme based on the dual of any minimal code
given in Theorem $\ref{thm:1}$, $\ref{thm:8}$, $\ref{thm:6}$  has interesting access structures. Since the dual distance $d$ of the constructed codes are all satisfying $d\geq 3$, the access structure of the corresponding secret sharing scheme is  democratic in the sense that every participant plays the same role in every decision making. On the other hand, the minimal access structures vary in size since the dimension of our constructed codes varies. This makes our secret sharing schemes more flexible and have much more interesting applications.

\section{Conclusion}

In this paper, we generalized a construction of a family of linear codes from almost bent functions and planar functions by making the dimension of them more flexible. We determined the parameters of all the constructed codes as well as parameters of their related codes such as dual codes, extended codes, duals of extended codes. A large family of optimal codes and almost optimal codes were given by our construction. As it was shown in Examples \ref{exam:1} and \ref{exam:2}, we obtained much more good linear codes with flexible dimensions.

Moreover, we found out that many of our constructed codes are minimal codes, and can be used to construct secret sharing schemes. What motivated us most is that the secret sharing schemes obtained from our constructed codes are more flexible and interesting, and can be used in much more applications. All of these show that the study of this paper is worthwhile and well motivated.


\begin{thebibliography}{10}
 \bibliographystyle{plain}

\bibitem{Anderson} R. Anderson, C. Ding, T. Helleseth, and T. Kl{\o}ve, How to build robust shared control systems,
Des., Codes, Cryptogr., vol. 15, no. 2, pp. 111-124, 1998.

\bibitem{Ashikhmin} A. Ashikhmin, A. Barg, Minimal vectors in linear codes,  IEEE Trans. Inf. Theory, vol. 44, no. 5, pp. 2010-2017, 1998.

\bibitem{AssmusM74} E. F. Assmus Jr., H. F. Mattson Jr., Coding and combinatorics, SIAM Rev. vol. 16, no. 3,  pp. 349-388, 1974.

\bibitem{Calderbank84} A. R. Calderbank, J. M. Goethals, Three-weight codes and association schemes, Philips
J. Res., vol. 39, pp. 143-152, 1984.

 \bibitem{Calderbank86} A. R. Calderbank, W. M. Kantor, The geometry of two-weight codes, Bull. London
Math. Soc., vol. 18, no. 2,  pp. 97-122, 1986.

\bibitem{Carlet2004} C. Carlet, C. Ding, Highly nonlinear mappings, J. Complexity, vol. 20, no. 2, pp. 205-244, 2004.

\bibitem{CarletDing05} C. Carlet, C. Ding, and J. Yuan, Linear codes from perfect nonlinear mappings and their secret sharing schemes, IEEE Trans. Inf. Theory, vol. 51, no. 6, pp. 2089-2102, 2005.

\bibitem{Coulter2002} R. S. Coulter, The number of rational points of a class of Artin-Schreier curves, Finite Fields Their Appl., vol. 8, pp. 397-413, 2002.

\bibitem{Coulter1997} R. S. Coulter,  R. W. Matthews, Planar functions and planes of Lenz-Barlotti class II, Des., Codes, Cryptogr., vol. 10, pp. 167-184, 1997.


\bibitem{Dembowski1968} P. Dembowski, T. G. Ostrom, Planes of order $n$ with collineation group of order $n^2$, Math. Z., vol. 193, pp. 239-258, 1968.

\bibitem{Dingbook} C. Ding, Designs from linear codes, World Scientific Publishing Co. Pte. ltd., 2019.

\bibitem{Ding2015} C. Ding, Linear codes from some 2-Designs, IEEE Trans. Inf. Theory, vol. 61, no. 6, pp. 3265-3275, 2015.

\bibitem{Ding2017} C. Ding, C. Li, Infinite families of 2-designs and 3-designs from linear codes, Discrete Math., vol. 340, no. 10, pp. 2415-2431,  2017.

\bibitem{Dingwang2005} C. Ding, X. Wang, A coding theory construction of new systematic authentication codes,
Theoretical Comput. Sci., vol. 330, no. 1, pp. 81-99, 2005.

\bibitem{Ding2007} C. Ding, J. Yin, Signal sets from functions with optimimum nonlinearity,  IEEE Trans. Inf. Communications, vol. 55, no. 5, pp. 4245-4250, 2007.

\bibitem{DingYuan} C. Ding, J. Yuan, A family of skew Hadamard difference sets, J. Combin. Theory, vol. 113, no. 7, pp. 1526-1535, 2006.

\bibitem{Dingsecret} C. Ding, J. Yuan, Covering and secret sharingwith linear codes, International Conference on Discrete Mathematics and Theoretical Computer Science, Springer, Berlin, Heidelberg, pp. 11-25, 2003.

\bibitem{kDing2015} K. Ding, C. Ding, A Class of two-weight and three-weight codes and their applications in
secret sharing, IEEE Trans. Inf. Theory, vol. 61, no. 11, pp. 5835-5842, 2015.

\bibitem{Fengluo07} K. Feng, J. Luo, Value distributions of exponential sums from perfect nonlinear functions and their applications, IEEE Trans.
Inf. Theory, vol. 53, no. 9, pp. 3035-3041, 2007.

\bibitem{Gold1968} R. Gold, Maximal recursive sequences with 3-valued recursive crosscorrelation functions,  IEEE Trans. Inf. Theory, vol. 14, no. 3, pp. 154-156, 1968.

\bibitem{Huffman03} W. C. Huffman,  V. Pless, Fundamentals of error-correcting codes, Cambridge university press, Cambridge, 2003.


\bibitem{Van99} Van Lint, Jacobus Hendricus, Introduction to coding theory,  Vol. 86, Springer Science and Business Media, 2012.


\bibitem{Mesnager2015code} S. Mesnager, Linear codes with few weights from weakly regular bent functions based on a
generic construction, Cryptogr. Communications, vol. 9, no. 9, pp. 71-84, 2017.

\bibitem{Rothaus1976} O. S. Rothaus, On bent functions, J. Combin. Theory Ser. A, vol. 20, no. 3, pp. 300-305, 1976.

\bibitem{Tang2016}  C. Tang, N. Li, Y. Qi, Z. Zhou, and T. Helleseth, Linear codes with two or three weights from
weakly regular bent functions, IEEE Trans. Inf. Theory, Vol. 62, No. 3, pp. 1166-1176, 2016.

\bibitem{Wolfmann99} J. Wolfmann, Bent functions and coding theory, Difference Sets, Sequences and their Correlation
Properties, Springer, Dordrecht, 1999: 393-418.


\bibitem{DingXiang2015} C. Xiang, C. Ding, and S. Mesnager, Optimal codebooks from binary codes meeting the Levenshtein bound, IEEE Trans. Inf. Theory, vol. 61, no. 12, pp. 6526-6535, 2015.


\bibitem{DingYuan06} J. Yuan, C. Carlet and C. Ding, The weight distribution of a class of linear codes from perfect nonlinear functions, IEEE Trans. Inf. Theory, vol. 52, no. 2, pp. 712-717, 2006.

\bibitem{YUanDing2006} J. Yuan, C. Ding, Secret sharing schemes from three classes of linear codes, IEEE Trans.
Inf. Theory, vol. 52, no. 1, pp. 206-212, 2006.

\bibitem{Zhou2015} Z. Zhou, N. Li, C. Fan, and T. Helleseth, Linear codes with two or three weights from
quadratic bent functions, Des., Codes, Cryptogr., vol. 81, no. 2, pp. 283-295, 2016.










    \end{thebibliography}
\end{document}